\documentclass{comsoc2016}

\usepackage{amsmath}
\usepackage{amsthm}
\usepackage{amsfonts}
\usepackage{amssymb}
\usepackage{stmaryrd} 

\usepackage{color}
\usepackage{tikz}
\usetikzlibrary{calc}
\usepackage{gnuplot-lua-tikz}
\usepackage{pgffor}

\usepackage{booktabs}
\usepackage{tabularx}
\usepackage{xspace}

\newcommand{\Omit}[1]{}
\usepackage{color}

\definecolor{mybeautifulcyan}{rgb}{0.08,0.58,0.89}
\definecolor{michel}{rgb}{0.74, 0.2, 0.64} 

\newcommand{\R}{\mathbb{R}}
\newcommand{\Rplus}{\R^+}

\newcommand{\union}{\cup}

\newcommand{\var}[1]{#1}

\newcommand{\ens}[1]{\mathcal{#1}}
\newcommand{\vect}[1]{\overrightarrow{#1}}

\newcommand{\Def}{\buildrel\hbox{\tiny \textit{def}}\over =}

\DeclareMathOperator{\argmax}{argmax}

\newtheorem{definition}{Definition}
\newtheorem{prop}{Proposition}
\newtheorem{example}{Example}
\newtheorem{remark}{Observation}
\newtheorem{corollary}{Corollary}

\newcommand{\Ag}{A} 
\newcommand{\Ob}{O}
\newcommand{\ag}{i}
\newcommand{\agp}{j}
\newcommand{\ob}{\ell}
\newcommand{\obp}{m}
\newcommand{\obpp}{p}

\newcommand{\Agents}{{\ens{\Ag}}}
\newcommand{\maxAgents}{{N}}
\newcommand{\Objets}{{\ens{\Ob}}}
\newcommand{\maxObjets}{{M}}

\newcommand{\Poids}{W}
\newcommand{\poids}{W}

\newcommand{\share}{\pi}
\newcommand{\Shares}{\vect{\share}}
\newcommand{\sharep}{\rho}
\newcommand{\Sharesp}{\vect{\sharep}}
\newcommand{\shareq}{\tau}

\newcommand{\price}{p}
\newcommand{\Prices}{\vect{\price}}

\newcommand{\Ut}{u}     

\newcommand{\llb}{\llbracket}
\newcommand{\rrb}{\rrbracket}

\newcommand{\conp}{\textsf{coNP}\xspace}

\newcolumntype{L}{>{\raggedright\arraybackslash}X}
\newcolumntype{R}{>{\raggedleft\arraybackslash}X}
\newcolumntype{C}{>{\centering\arraybackslash}X}

\newcounter{problem}

\usepackage[english,vlined,ruled,linesnumbered]{algorithm2e}
\SetCommentSty{}
\SetFuncSty{sc}
\SetKwIF{Si}{SinonSi}{Sinon}{si}{alors}{sinon si}{sinon}{finsi}
\SetKwFor{Pour}{pour}{faire}{finpour}
\SetKwFor{Tq}{tant que}{faire}{fintq}
\SetKwInput{Donnes}{Données}
\SetKwInput{Res}{Résultat}
\SetKwInput{Input}{Input}
\SetKwInput{Output}{Output}
\SetKw{Retour}{renvoyer}
\SetKw{From}{from}
\SetKw{To}{to}
\SetKw{Choose}{choose}
\makeatletter
\renewcommand{\@algocf@procname}{Procédure}
\makeatother

\newcommand{\SI}{{\cal S}(I)}
\newcommand{\PI}{{\cal P}(I)}
\newcommand{\instance}{\langle \Agents, \Objets, \poids \rangle}
\newcommand{\sequence}{\sigma} 
\newcommand{\seq}{\vect{\sequence}} 
\newcommand{\seqi}{\sequence}
\newcommand{\best}{\mathrm{best}}

\newcommand{\soscs}{sequences of sincere choices\xspace}
\newcommand*{\al}[1]{
\begin{tikzpicture}
[baseline=(letter.base)]
[every text node part/.style={font=\rm}],
\node[draw,circle,inner sep=1pt] (letter) {$\mathrm #1$};
\end{tikzpicture}
}

\title{Efficiency and Sequenceability in Fair Division of Indivisible Goods with Additive Preferences}
\author{Sylvain Bouveret and Michel Lema\^itre}

\begin{document}


\begin{abstract}
  In fair division of indivisible goods, using sequences of sincere
  choices (or picking sequences) is a natural way to allocate the
  objects. The idea is the following: at each stage, a designated
  agent picks one object among those that remain.  This paper,
  restricted to the case where the agents have numerical additive
  preferences over objects, revisits to some extent the seminal paper
  by Brams and King \cite{BramsKing05} which was specific to ordinal
  and linear order preferences over items. We point out similarities
  and differences with this latter context.  In particular, we show
  that any Pareto-optimal allocation (under additive preferences) is
  sequenceable, but that the converse is not true anymore.  This
  asymmetry leads naturally to the definition of a ``scale of
  efficiency'' having three steps: Pareto-optimality, sequenceability
  without Pareto-optimality, and non-sequenceability.  Finally, we
  investigate the links between these efficiency properties and the
  ``scale of fairness'' we have described in an earlier work
  \cite{Bouveret2015CharacterizingConflicts}: we first show that an
  allocation can be envy-free and non-sequenceable, but that every
  competitive equilibrium with equal incomes is sequenceable. Then we
  experimentally explore the links between the scales of efficiency
  and fairness.
\end{abstract}


\section{Introduction}

In this paper, we investigate fair division of indivisible goods. In
this problem, a set of indivisible objects or goods has to be
allocated to a set of agents (individuals, entities...), taking into
account, as best as possible, the agents' preferences about the
objects. This classical collective decision making problem has plenty
of practical applications, among which the allocation of space
resources \cite{Lemaitre99,Bianchessi07}, of tasks to workers in
crowdsourcing market systems, papers to reviewers
\cite{goldsmith2007ai} or courses to students \cite{Budish11}.

This problem can be tackled from two different perspectives. The first
possibility is to resort to a benevolent entity in charge of
collecting in a \emph{centralized} way the preferences of all the
agents about the objects. This entity then computes an allocation that
takes into account these preferences and satisfies some fairness
(\textit{e.g.} envy-freeness) and efficiency (\textit{e.g.}
Pareto-optimality) criteria, or optimizes a well-chosen social welfare
ordering.  The second possibility is to have a \emph{distributed}
point of view, \textit{e.g.} by starting from an initial allocation
and letting the agents negotiate to swap their objects
\cite{Sandholm98,Chevaleyre05ijcai}.

A somewhat intermediate approach consists in allocating the objects to
the agents using a \emph{protocol}, which can be seen as a
way of building an allocation interactively by asking the agents a
sequence of questions. Protocols are at the heart of works mainly
concerning the allocation of divisible resources (cake-cutting) --- see
Brams and Taylor's seminal book \cite{Brams96} for a reference --- but
have also been studied in the context of indivisible goods
\cite{Brams96,brams2012undercut}.

In this paper, we focus on a particular allocation protocol:
\emph{sequences of sincere choices} (also known as \emph{picking
  sequences}). This very simple and natural protocol works as follows.
A central authority chooses a sequence of agents before the protocol
starts, having as many agents as the number of objects (some agents
may appear several times in the sequence). Then, each agent appearing
in the sequence is asked to choose 
in turn one object among
those that remain. For instance, according to the sequence
$\langle 1,2,2,1 \rangle$, agent $1$ is going to choose first, then
agent $2$ will pick two consecutive objects, and agent $1$
will take the last object. This simple protocol, actually used in a
lot of everyday situations,\footnote{It is for instance
  used in the board game \textit{The Settlers of Catane} for
  allocating initial resources.} has been studied for the first time
by Kohler and Chandrasekaran \cite{KC71}. Later, Brams and Taylor
\cite{Brams00} have studied a particular version of this protocol,
namely alternating sequences, in which the sequence of agents is
restricted to a balanced ($\langle 1,2,2,1...\rangle$) or strict
($\langle 1,2,1,2...\rangle$) alternation of agents. Bouveret and Lang
\cite{BL-IJCAI11} have further formalized this protocol, whose
properties (especially related to game theoretic aspects) have been
characterized by Kalinowski {\em et al.}
\cite{KalinowskiNW13,KalinowskiNWX13}. Finally, Aziz \emph{et al.}
\cite{Aziz2015possible} have studied the complexity of problems
related to finding whether a particular assignment (or bundle)
is achievable by a particular class of picking sequences.

On top of all these works specifically dedicated to this
protocol, we can add the interesting work by Brams \textit{et al.}
\cite{BramsKing05}. This work, which focuses on a situation where
the agents have ordinal preferences, is not specifically dedicated
to picking sequences. However, the authors make an interesting link
between this protocol and Pareto-optimality, showing, among
others, that picking sequences always result in a Pareto-optimal
allocation, but also that every Pareto-optimal allocation can be
obtained in this way.

In this paper, we will elaborate on these ideas and analyze the links
between sequences and some efficiency and fairness properties, in a
more general model in which the agents have numerical additive
preferences on the objects, without any further restriction. Our main
contributions are the following.  We give a formalization of the link
between allocation and \soscs, 
highlighting a simple characterization of
the sequenceability of an allocation. Then, we show that in this
slightly more general framework than the one by Brams \emph{et al.},
surprisingly, Pareto-optimality and sequenceability are not equivalent
anymore. As a consequence we can define a ``scale of efficiency'' that
allows us to characterize the degree of efficiency of a given
allocation. We also highlight an interesting link between
sequenceability and another important economical concept: the
competitive equilibrium from equal income (CEEI). Another contribution
is the experimental exploration of the links between the scale of
efficiency and fairness properties, which has led us to develop, among
others, a practical method for testing if a given allocation passes
the CEEI test --- a practical problem which was still open
\cite{Bouveret2015CharacterizingConflicts}.

The paper is organized as follows. Section~\ref{sec:modele} describes
the model of our allocation problem and the allocation protocol based
on \soscs. Section~\ref{sec:seqAlloc} focuses on the problem
consisting in deciding whether an allocation can be obtained by a
sequence of sincere choices. In the next two sections, we analyze the
relation between sequences of sincere choices and three classical
properties: Pareto-optimality (Section~\ref{sec:Pareto}),
envy-freeness and competitive equilibrium from equal income
(Section~\ref{sec:CEEI}). Finally, we explore experimentally in
Section~\ref{sec:expe} the links between the ``scale of efficiency''
mentioned above, and the ``scale of fairness'' that we have
described in a previous work \cite{Bouveret2015CharacterizingConflicts}.


\section{Model and Definitions}
\label{sec:modele}

The aim the fair division of indivisible goods, also called MultiAgent
Resource Allocation (MARA), is to allocate a finite set of
\emph{objects} $\Objets = \{1,\dots,\maxObjets\}$ to a finite set of
\emph{agents} $\Agents = \{1,\dots,\maxAgents\}$. A
\emph{sub-allocation} on $\Objets' \subseteq \Objets$ is a vector
$\Shares^{|\Objets'} = \langle
\share_1^{|\Objets'},\dots,\share_{\maxAgents}^{|\Objets'}\rangle$,
such that
$\ag \neq \agp
\Rightarrow\share_\ag^{|\Objets'}\cap\share_\agp^{|\Objets'}=\emptyset$
(a given object cannot be allocated to more than one agent) and
$\union_{\ag\in\Agents} \share_\ag^{|\Objets'} = \Objets'$ (all the
objects from $\Objets'$ are allocated).
$\share_\ag^{|\Objets'} \subseteq \Objets'$ is called agent $\ag$'s
\emph{share} on $\Objets'$.  $\Shares^{|\Objets''}$ is a
sub-allocation of $\Shares^{|\Objets'}$ when
$\share^{|\Objets''}_\ag \subseteq \share^{|\Objets'}_\ag$ for each
agent $\ag$. Any sub-allocation $\Shares^{|\Objets} $ on the entire
set of objects will be denoted $\Shares $ and just called
\emph{allocation}.

Any satisfactory allocation must take into account the agents'
preferences on the objects. Here, we will make the classical
assumption that these preferences are \emph{numerically additive}.
Each agent $\ag$ has a \emph{utility function} $\Ut_\ag:2^\Objets \to \Rplus$
measuring her satisfaction $\Ut_\ag(\share)$ when she obtains share $\share$,
which is defined as follows:
\begin{equation*} \label{eq:utilite} \Ut_\ag(\share) \Def
  \sum_{\ob\in\share}\poids(\ag,\ob),
\end{equation*}
where $\poids(\ag, \ob)$ is the weight given by agent $\ag$ to object
$\ob$. This assumption, as restrictive as it may seem, is made by a
lot of authors \cite[for instance]{Lipton04,Bansal06} and is
considered a good compromise between expressivity and conciseness.

If we put things together:
\begin{definition}
  An instance of the \emph{additive multiagent resource allocation
    problem} (add-MARA instance for short) $I = \instance$ is a tuple
  $(\Agents, \Objets, \poids)$, where:
  \begin{itemize}
  \item $\Agents = \{1, \dots ,\ag, \dots, \maxAgents\}$ is a set of
    $\maxAgents$ \emph{agents};
  \item $\Objets = \{1, \dots ,\ob, \dots \maxObjets\}$ is a set of
    $\maxObjets$ \emph{objects},
  \item $\poids: \Agents \times \Objets \to \Rplus$ is a mapping,
    $\poids(\ag, \ob)$ being the weight given by agent $\ag$ to
    object $\ob$.
  \end{itemize}
  We will denote by $\PI$ the set of allocations for $I$.
\end{definition}

We will say that \emph{the agents' preferences are strict on the
  objects} if, for each agent $\ag$ and each pair of objects
$\ob \neq\obp$, we have $\poids(\ag, \ob) \neq \poids(\ag, \obp)$.
Similarly, we will say that \emph{the agents' preferences are strict
  on the shares} if, for each agent $\ag$ and each pair of shares
$\share\neq\share'$, we have $\Ut_\ag(\share) \neq \Ut_\ag(\share')$.
Finally, we will say that \emph{the agents have same order
  preferences} \cite{Bouveret2015CharacterizingConflicts} if there is
a permutation $\eta : \Objets \mapsto \Objets$ such that for each
agent $\ag$ and each pair of objects $\ob$ and $\obp$, if
$\eta(\ob) < \eta(\obp)$ then
$\poids(\ag, \eta(\ob)) \geq \poids(\ag, \eta(\obp))$.

\begin{remark} \label{rem:str} Stricticity on shares implies
  stricticity on objects, but the converse is false.
\end{remark}

The following definition will play a prominent role.

\begin{definition}
  Given an agent $\ag$ and a set of objects $O$, let
  $\best(O, \ag) = \argmax_{\ob\in O} \poids(\ag, \ob)$ be the subset 
  of objects in $O$ having the highest weight for agent $i$ (such
  objects will be called \emph{top} objects of $i$).  A
  (sub-)allocation $\Shares^{|\Objets'}$ is said \emph{frustrating} if
  no agent receives one of her top objects in $\Shares^{|\Objets'}$
  (formally:
  $\best(\Objets', \ag) \cap \Shares^{|\Objets'}_\ag = \emptyset$ for
  each agent $\ag$), and \emph{non-frustrating} if at least one agent
  receives one of her top objects in $\Shares^{|\Objets'}$.
\end{definition}

It should be emphasized that this notion of frustrating allocation was
already present but implicit in the work by Brams and King
\cite{BramsKing05}. Here, we bring this concept out because it will
lead to a nice characterization of sequenceable allocations, as we
will see later.

\medskip

In the following, we will consider a particular way of allocating
objects to agents: allocation by sequences of sincere
choices. Informally the agents are asked in turn, according to a
predefined sequence, to choose and pick a top object among the
remaining ones.

\begin{definition} Let $I = \instance$ be an add-MARA instance. A
  \emph{sequence of sincere choices} (or simply \emph{sequence} when
  the context is clear) is a vector of $\Agents^\maxObjets$.
We will denote by $\SI$ the set of possible sequences for instance $I$.
\end{definition}

Let $\seq \in \SI$. $\seq$ is said to \emph{generate} allocation
$\Shares$ if and only if $\Shares$ can be obtained as a possible
result of the non-deterministic\footnote{The algorithm contains an
  instruction $\Choose$ splitting the control flow into several
  branches, building all the allocations generated by $\seq$.}
Algorithm~\ref{algo:ss} on input $I$ and $\seq$.

\begin{algorithm}
  \caption{Execution of a sequence}
  \label{algo:ss}
  \Input{an instance $I = \instance$ and a sequence $\seq \in \SI$}
  \Output{an allocation $\Shares \in \PI$}

  $\Shares \gets$ empty allocation (such that $\forall \ag \in \Agents : \share_\ag = \emptyset$)\;
  $\Objets_1 \gets \Objets$\;
  \For{$t$ \From 1 \To $\maxObjets$}{
    $\ag \gets \seqi_t$\;
    \Choose object $o_t \in \best(\Objets_t, \ag)$\label{line:guess} \;
    $\share_\ag \gets \share_\ag  \union \{ o_t \} $ \;
    $\Objets_{t+1} \gets \Objets_t \setminus  \{ o_t \} $
  }
\end{algorithm}

\begin{definition} \label{def:sequenceable} An allocation $\Shares$ is
  said to be \emph{sequenceable} if there exists a sequence $\seq$
  that generates $\Shares$, and \emph{non-sequenceable} otherwise. For
  a given instance I, we will denote by $s(I)$ the binary relation
  defined by $(\seq,\Shares)\in s(I)$ if and only if $\Shares$ can be
  generated by $\seq$.
\end{definition}

\begin{example} \label{ex:A1} Let $I$ be the instance represented by
  the following weight matrix:\footnote{In this example an in the
    following ones, we will represent instances by a matrix in which
    each value at row $\ag$ and column $\ob$ represents the weight
    $W(\ag, \ob)$. We will also use notation $ab...$ as a shorthand
    for $\{a, b, ...\}$}
  \[  \left(
    \begin{array}{rrr}
      8 & 2 & 1 \\
      5 & 1 & 5 
    \end{array}
  \right) \]
  The binary relation $s(I)$ between $\SI$ and $\PI$ can be graphically
  represented as follows:

  
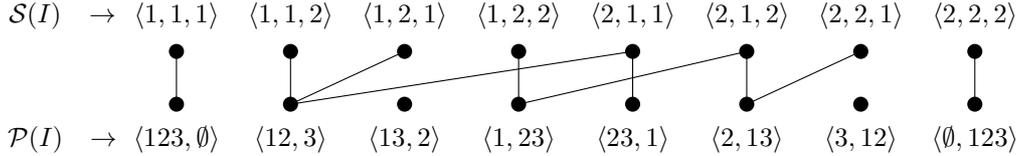
\vspace{-0.2cm}  
  \begin{center}
    \begin{tikzpicture}
      \draw (-1.5, 0) node[circle, inner sep=2pt, outer sep=5pt] (node-S1) {} node[above, outer sep=5pt] {$\SI$\quad$\to$};

      \draw (0, 0) node[circle, fill=black, inner sep=2pt, outer sep=5pt] (node-a) {} node[above, outer sep=5pt] {$\langle 1, 1, 1\rangle$};
      \draw (1.5, 0) node[circle, fill=black, inner sep=2pt, outer sep=5pt] (node-b) {} node[above, outer sep=5pt] {$\langle 1, 1, 2\rangle$};
      \draw (3, 0) node[circle, fill=black, inner sep=2pt, outer sep=5pt] (node-c) {} node[above, outer sep=5pt] {$\langle 1, 2, 1\rangle$};
      \draw (4.5, 0) node[circle, fill=black, inner sep=2pt, outer sep=5pt] (node-d) {} node[above, outer sep=5pt] {$\langle 1, 2, 2\rangle $};
      \draw (6, 0) node[circle, fill=black, inner sep=2pt, outer sep=5pt] (node-e) {} node[above, outer sep=5pt] {$\langle 2, 1, 1\rangle$};
      \draw (7.5, 0) node[circle, fill=black, inner sep=2pt, outer sep=5pt] (node-f) {} node[above, outer sep=5pt] {$\langle 2, 1, 2\rangle$};
      \draw (9, 0) node[circle, fill=black, inner sep=2pt, outer sep=5pt] (node-g) {} node[above, outer sep=5pt] {$\langle 2, 2, 1\rangle$};
      \draw (10.5, 0) node[circle, fill=black, inner sep=2pt, outer sep=5pt] (node-h) {} node[above, outer sep=5pt] {$\langle 2, 2, 2\rangle$};

      \draw (-1.5, -0.7) node[circle, inner sep=2pt, outer sep=5pt] (node-P1) {} node[below, outer sep=5pt] {$\PI$\quad$\to$};

      \draw (0, -0.7) node[circle, fill=black, inner sep=2pt, outer sep=5pt] (node-A) {} node[below, outer sep=5pt] {$\langle  123, \emptyset \rangle$};
      \draw (1.5, -0.7) node[circle, fill=black, inner sep=2pt, outer sep=5pt] (node-B) {} node[below, outer sep=5pt] {$\langle  12 , 3 \rangle$};
      \draw (3, -0.7) node[circle, fill=black, inner sep=2pt, outer sep=5pt] (node-C) {} node[below, outer sep=5pt] {$\langle  13, 2 \rangle$};
      \draw (4.5, -0.7) node[circle, fill=black, inner sep=2pt, outer sep=5pt] (node-D) {} node[below, outer sep=5pt] {$\langle  1 , 23 \rangle$};
      \draw (6, -0.7) node[circle, fill=black, inner sep=2pt, outer sep=5pt] (node-E) {} node[below, outer sep=5pt] {$\langle  23 , 1 \rangle$};
      \draw (7.5, -0.7) node[circle, fill=black, inner sep=2pt, outer sep=5pt] (node-F) {} node[below, outer sep=5pt] {$\langle 2 ,13 \rangle$};
      \draw (9, -0.7) node[circle, fill=black, inner sep=2pt, outer sep=5pt] (node-G) {} node[below, outer sep=5pt] {$\langle  3 , 12 \rangle$};
      \draw (10.5, -0.7) node[circle, fill=black, inner sep=2pt, outer sep=5pt] (node-H) {} node[below, outer sep=5pt] {$\langle  \emptyset , 123 \rangle$};

      \foreach \x/\y in {a/A, b/B, c/B, d/D, e/B,  e/E, f/D, f/F, g/F, h/H}
      \draw (node-\x.center) -- (node-\y.center);
    \end{tikzpicture}
  \end{center}

  \vspace{-0.3cm}
  For instance, sequence $\langle 2, 1, 2 \rangle$ generates two
  possible allocations: $\langle 1, 23 \rangle$ and
  $\langle 2, 13 \rangle$, depending on whether agent 2 chooses object
  1 or 3 that she both prefers. Allocation $\langle 12, 3 \rangle$ can
  be generated by three sequences. Allocations $\langle 13, 2 \rangle$
  and $\langle 3, 12 \rangle$ are non-sequenceable.
\end{example}

Once again, this notion of sequenceability is already implicitly
present in the work by Brams and King \cite{BramsKing05}, and has been
extensively studied by Aziz \emph{et al.}
\cite{Aziz2015possible}. However, a fundamental difference is that in
our setting, the preferences might be non strict on objects, which
entails that the same sequence can yield different allocations (in the
worst case, an exponential number), as Example~\ref{ex:A1} shows.
 
\begin{remark} \label{rem:card} For any instance $I$,
  $|\SI| = |\PI| = \maxAgents^\maxObjets$.
\end{remark}

\begin{remark} The number of objects allocated to an agent by a
  sequence is equal to the number of times the agent appears in the
  sequence. Formally: for all $(\seq,\Shares)\in s(I)$ and all
  agent $\ag$,
  $|\share_\ag| = \sum_{\ob\in\Objets}[\seqi_\ob = \ag]$ where
  $[z = t]$ is 1 if the equality is verified, and 0 otherwise.
\end{remark}


\section{Sequenceable allocations}
\label{sec:seqAlloc}

In this section and in the following one, we will give a
characterization of sequenceable allocations, that is, we will try to
identify under which conditions an allocation is achievable by the
execution of a sequence of sincere choices. The question has already
been extensively studied by Aziz \emph{et al.}
\cite{Aziz2015possible}, but in a quite different context --- namely,
ordinal strict preferences on objects --- and with a particular focus
on sub-classes of sequences (\textit{e.g.}  alternating sequences). As
we will show, the properties are not completely similar in our
context.

\subsection{Characterization}
 
We have seen in Example~\ref{ex:A1} that some allocations are
non-sequenceable. We will now formalize this and give a precise
characterization of sequenceable allocations. We first start by
noticing that in every sequenceable allocation, the first agent of the
sequence gets a top object, which yields the following remark:

\begin{remark} \label{rem:nonSeq} Every frustrating allocation is
  non-sequenceable.
\end{remark}
 
\begin{example} \label{ex:B} In the following instance,
  the circled allocation $\langle 23, 1 \rangle$ is
  non-sequenceable because it is frustrating.
  \[  \left(
    \begin{array}{ccc}
      2 & \al 1  & \al 1\\
      \al 1 & 2 & 2
    \end{array}
  \right) \]
\end{example}

However, it is possible to find a non-sequenceable allocation that
gives her top object to one agent (as allocation
$\langle 13, 2 \rangle$ in Example~\ref{ex:A1}) or even to all, as the
following example shows.

\begin{example} \label{ex:C} Consider this instance: 
  \[  \left(
    \begin{array}{cccc}
      \al 9 & 8 & \tikz[remember picture]\node[inner sep=0pt] (ul1) {};2 & \al 1 \\
      2 & \al 5 & \al 1 & 4\tikz[remember picture]\node[inner sep=0pt] (br1) {}; 
    \end{array}
  \right) \]

  \tikz[remember picture]\draw[overlay,dotted] ($ (ul1) + (-0.2, 0.4) $) rectangle ($ (br1) + (0.2, -0.15) $);

  In the circled allocation
  $\Shares=\langle 14, 23 \rangle$, every agent receives her
  top object. However, after objects 1 and 2 have been allocated (they
  must be allocated first by all sequence generating $\Shares$), the
  sub-allocation shown in a dotted box above remains.  This
  sub-allocation is obviously non-sequenceable because it is
  frustrating. Hence $\Shares$ is not sequenceable either.
\end{example}

This property of containing a frustrating sub-allocation exactly
characterizes the set of non-sequenceable allocations:
 
\begin{prop}\label{prop:sousAlloc}
  Let $I = \instance$ be an instance and $\Shares$ be an allocation of
  this instance. The two following statements are equivalent:
  \begin{itemize}
  \item[(A)] $\Shares$ is sequenceable.
  \item[(B)] No sub-allocation of $\Shares$ is frustrating (in every
    sub-allocation, at least one agent receives a top object).
  \end{itemize}
\end{prop}

\begin{proof} (B) implies (A). Let us suppose that for all subset of
  objects $\Objets' \subseteq \Objets$ there is at least one agent
  obtaining one of her top objects in $\Shares^{|\Objets'}$. We will
  constructively show that $\Shares$ is sequenceable. Let $\seq$ be a
  sequence of agents and $\vect{\Objets} \in (2^\Objets)^\maxObjets$
  be a sequence of sets of objects jointly defined as follows:
  \begin{itemize}
  \item $\Objets_1 = \Objets$ and $\seqi_1$ is an agent that receives
    one of her top objects in $\Shares^{|\Objets_1}$;
  \item $\Objets_{t+1} = \Objets_{t} \setminus \{o_{t}\}$, where
    $o_{t} \in \best(\Objets_{t}, \seqi_{t})$ and $\seqi_t$ is an
    agent that receives one of her top objects in
    $\Shares^{|\Objets_t}$, for $t \geq 1$.
  \end{itemize}
  
  From the assumption on $\Shares$, we can check that the sequence
  $\seq$ is perfectly defined. Moreover, $\Shares$ is one of the
  allocations generated by $\seq$.
  
  (A) implies (B) by contraposition. Let $\Shares$ be an allocation
  containing a frustrating sub-allocation
  $\Shares^{|\Objets'}$. Suppose that there exists a sequence $\seq$
  generating $\Shares$.  We can notice that in
  Algorithm~\ref{algo:ss}, when an object is allocated to an agent,
  all the objects which are strictly better for her have already been
  allocated at a previous step. Let $\ob\in\Objets'$, and let $\ag$ be
  the agent that receives $\ob$ in $\Shares$. Since
  $\Shares^{|\Objets'}$ is frustrating, there is another object
  $\obp \in \Objets'$ such that $W(\ag, \obp) > W(\ag, \ob)$. From the
  previous remark, $\obp$ is necessarily allocated before $\ob$ in the
  execution of sequence $\seq$.  We can deduce, from the same line of
  reasoning on $\obp$ and agent $\agp$ that receives it, that there is
  another object $\obpp$ allocated before $\obp$ in the execution of
  the sequence.  The set $\Objets'$ being finite, using the same
  argument iteratively, we will necessarily find an object which has
  already been encountered before. This leads to a cycle in the
  precedence relation of the objects in the execution of the
  sequence. Contradiction: no sequence can thus generate $\Shares$.
  \end{proof}

  Beyond the fact that it characterizes a sequenceable allocation, the proof of
  Proposition~\ref{prop:sousAlloc} gives
  a practical way of checking if an allocation is sequenceable, and,
  if it is the case, of computing a sequence that generates this
  allocation. This yields the following result:

  \begin{prop}
    Let $I = \instance$ be an instance and $\Shares$ be an allocation
    of this instance. We can decide in time
    $O(\maxAgents \times \maxObjets^2)$ if $\Shares$ is sequenceable.
  \end{prop}

  The proof of this proposition is based on the execution of a similar
  algorithm than the one which is used by Brams and King
  \cite{BramsKing05} in the proof of their Proposition 1
  (necessity). However, our algorithm is more general because (i) it
  can involve non-strict preferences on objects, and (ii) it can
  conclude with non-sequenceability.

  \begin{proof} 
    We show that Algorithm~\ref{algo:testSeq} returns a sequence
    $\seq$ generating the input allocation $\Shares$ if and only if
    there is one. Suppose that the algorithm returns a sequence
    $\seq$. Then, by definition of the sequence (in the loop from line
    \ref{lgn:defSeq1} to line \ref{lgn:defSeq2}), at each step $t$,
    $\ag =\seqi_t$ can choose an object in $\share_i$, that is one of
    her top objects. Conversely, suppose the algorithm returns
    \textit{Non-sequenceable}. Then, at a given step $t$,
    $\forall \ag$, $\best(\Objets', \ag) \cap \share_\ag = \emptyset$.
    By definition, $\Shares^{|\Objets'}$ is therefore, at this step, a
    frustrating sub-allocation of $\Shares$. By
    Proposition~\ref{prop:sousAlloc}, $\Shares$ is thus
    non-sequenceable.
    The loop from line \ref{lgn:defSeq1} to line \ref{lgn:defSeq2}
    runs in time $O(\maxAgents \times \maxObjets)$, because searching
    for the top objects in the preferences of each agent can be made
    in time $O(\maxObjets)$. This loop being executed $\maxObjets$
    times, the algorithm globally runs in time
    $O(\maxAgents \times \maxObjets^2)$.
    \begin{algorithm}
      \caption{Sequencing of an allocation}
      \label{algo:testSeq}
      \Input{an instance $I = \instance$ and an allocation $\Shares \in \PI$}
      \Output{a sequence $\seq$ generating $\Shares$ or \textit{Non-sequenceable}}
      
      $(\seq, \Objets') \gets (\langle\rangle, \Objets)$\;
      
      \For{$t$ \From 1 \To $\maxObjets$\label{lgn:defSeq1}}{
        \If{$\exists \ag$ such that $\best(\Objets', \ag) \cap \share_\ag \neq \emptyset$}{
          $\seq \gets \seq \cdot \ag$\;
          let $\ob \in \best(\Objets', \ag) \cap \share_\ag$\;
          $\Objets' \gets \Objets' \setminus \{\ob\}$\;
        }
        \lElse{
          \Return{\textit{Non-sequenceable}\label{lgn:defSeq2}}
        }
      }
      \Return $\seq$\;
    \end{algorithm}
\end{proof}

\subsection{Strict preferences on objects}

We now characterize the instances for which the relation $s(I)$
is an application.

\begin{prop} \label{prop:StrObj} Preferences are strict on objects if
  and only if the relation $s(I)$ is an application from $\SI$
  to $\PI$.
\end{prop}

\begin{proof}  
  If preferences are strict on objects, then each agent has only one
  possible choice at line~\ref{line:guess} of Algorithm~\ref{algo:ss}
  and hence every sequence generates one and only one allocation.

  Conversely, if preferences are not strict on objects, at least one
  agent (suppose w.l.o.g. agent 1) gives the same weight to two
  different objects. Suppose that there is at least $t$ objects ranked
  above. Then obviously, the following sequence
  $ \underbrace{111...111}_{t + 1\text{
      times}}\underbrace{222...222}_{\maxObjets - t - 1\text{ times}}
  $
  generates two allocations, depending on agent 1's choice at step
  $t+1$.
\end{proof}


\subsection{Same order preferences}

\begin{prop} \label{prop:poi} All the allocations of an instance with
  same order preferences are sequenceable. Conversely, if all the
  allocations of an instance are sequenceable, then this instance has
  same order preferences.
\end{prop}

\begin{proof} Let $I$ be an instance with same order preferences, and
  let $\Shares$ be an arbitrary allocation.  In every sub-allocation
  of $\Shares$ at least one agent obtains a top object (because the
  preference order is the same among agents) and hence cannot be
  frustrating. By Proposition~\ref{prop:sousAlloc}, $\Shares$ is
  sequenceable.

       Conversely, let us assume for contradiction that $I$ is an instance
  not having same order preferences. Then there are two distinct
  objects $\ob$ and $\obp$ and two distinct agents $\ag$ and $\agp$
  such that $\Poids(\ag,\ob) \geq \Poids(\agp,\ob)$ and
  $\Poids(\ag,\obp) \leq \Poids(\agp,\obp)$, one of the two
  inequalities being strict (assume w.l.o.g. the first one). The
  sub-allocation $\Shares^{|\{\ob, \obp\}}$ such that
  $\share^{|\{\ob, \obp\}}_\ag = \{\obp\}$ and
  $\share^{|\{\ob, \obp\}}_\agp = \{\ob\}$ is frustrating. By
  Proposition~\ref{prop:sousAlloc}, every allocation $\Shares$
  containing this frustrating sub-allocation (hence such that
  $\obp \in \share_\ag$ and $\ob \in \share_\agp$) is
  non-sequenceable.
\end{proof}

Let us now characterize the instances for which the relation
$s(I)$ is a bijection.

\begin{prop} \label{prop:SOP}
  For a given instance, the following two statements are equivalent.
  \begin{itemize}
  \item[(A)] Preferences are strict on objects and same order.
  \item[(B)] The relation $s(I)$ is a bijection.
  \end{itemize}
\end{prop}

The proof of this proposition is an easy consequence of
Propositions~\ref{prop:StrObj} and \ref{prop:poi}.




\section{Pareto-optimality}
\label{sec:Pareto}

An allocation is Pareto-optimal if there is no other allocation
dominating it. In our context, allocation $\Shares'$ dominates
allocation $\Shares$ if for all agent $\ag$,
$\Ut_\ag(\share_\ag') \geq \Ut_\ag(\share_\ag)$ and
$\Ut_\agp(\share_\agp') > \Ut_\agp(\share_\agp)$ for at least one
agent $\agp$. Pareto-optimality formalizes the idea of efficiency:
when an allocation is Pareto-optimal, one cannot strictly increase the
utility of one agent without strictly decreasing another one's.

When an allocation is generated from a sequence, in some sense, a weak
form of efficiency is applied to build the allocation: each successive
(picking) choice is ``locally'' optimal. This raises a natural
question: is every sequenceable allocation Pareto-optimal?

Brams \textit{et al.} \cite[Proposition 1]{BramsKing05} answer
positively by proving the equivalence between sequenceability and
Pareto-optimality.  However, they have a different notion of
Pareto-optimality, because they only have partial ordinal information
about the agents' preferences. More precisely, in Brams and King's
model, the agents' preferences are given as linear orders over
\emph{objects} (\textit{e.g.}  $1 \succ 2 \succ 3 \succ 4$). To be
able to compare bundles, these preferences are lifted on subsets using
the \emph{responsive set extension} $\succ_{RS}$, which is similar to
the one defined in the work by Aziz \emph{et al.}  \cite{Aziz15fair},
and to the one defined by Bouveret \emph{et al.}  \cite{BEL-ECAI10}
for SCI-nets.\footnote{It is actually not completely clear in Brams
  and King's paper whether or not their notion of dominance extends to
  bundles of different sizes, but it seems to be implicitly the case,
  using monotonicity.} This extension leaves many bundles incomparable
(\textit{e.g.} $14$ and $23$ if we consider the order
$1 \succ 2 \succ 3 \succ 4$). It leads Bouveret \emph{et al.}
\cite{BEL-ECAI10} to define, among others, two modes of
Pareto-optimality: \emph{possible} and \emph{necessary}
Pareto-optimality. Brams and King's notion of Pareto-optimality
exactly corresponds to possible Pareto-optimality.

Aziz \emph{et al.} \cite{Aziz15fair} show that, given a linear order
$\succ$ on objects and two bundles $\share$ and $\share'$,
$\share \succ_{RS} \share'$ if and only if $u(\share) > u(\share')$
for \emph{all} additive utility function $u$ compatible with $\succ$
(that is, such that $u(\ob) > u(\obp)$ if and only if
$\ob \succ \obp$). This characterization of responsive dominance
yields the following reinterpretation of Brams and King's result:

\begin{prop}[Brams and King \cite{BramsKing05}]
  Let $\langle \succ_1, \dots, \succ_\maxAgents \rangle$ be the
  profile of agents' ordinal preferences (represented as linear
  orders). Allocation $\Shares$ is sequenceable if and only if for
  each other allocation $\Shares'$, there is a sequence
  $u_1, \dots, u_\maxAgents$ of additive utility functions,
  respectively compatible with $\succ_1, \dots, \succ_\maxAgents$
  such that $u_\ag(\Shares) > u_\ag(\Shares')$ for at least one
  agent $\ag$.
\end{prop}

The latter notion of Pareto-optimality is very weak, because the
additive utility function is not fixed --- we just have to find one
that works. In our context where the utility function is fixed and
hence leads to a much stronger notion of Pareto-optimality, there is
no reason to suppose that Pareto-optimality is equivalent to
sequenceability anymore. And it turns out that it is indeed not the
case, as the following example shows.

\begin{example} \label{ex:E} Let us consider the following instance:
  \[ \left(
    \begin{array}{rrr}
      5 & 4 & 2 \\
      8 & 2 & 1 
    \end{array}
  \right) \]
  

  
  
  The sequence $\langle 1, 2, 2 \rangle$ generates allocation
  $A = \langle 1, 23 \rangle$ giving utilities
  $\langle 5, 3 \rangle$. $A$ is dominated by
  $B = \langle 23, 1 \rangle$, giving utilities $\langle 6, 8 \rangle$
  (and generated by $\langle 2, 1, 1 \rangle$). Observe that, under
  ordinal linear preferences, $B$ would not dominate $A$, but
  they would be incomparable.
\end{example}

The last example shows that a sequence of sincere choices does not
necessarily generate a Pareto-optimal allocation (even when the
preferences are same order and strict on shares, as the example
shows). What about the converse? We can see, as a trivial corollary of
the latter reinterpretation of Brams and King's result, that the
answer is positive \emph{if the preferences are strict on shares}. The
following result is more general:

\begin{prop} \label{prop:POseq} Every Pareto-optimal allocation is
  sequenceable.
\end{prop}

Before giving the formal proof, we illustrate it on a concrete example
\cite[Example 5]{Bouveret2015CharacterizingConflicts}.

\begin{example} \label{ex:EFnonPOprefsStrictes} 
 Let us consider the following instance:
  \[
  \Poids = \left(
    \begin{array}{ccccc}
      2 & 12 & \tikz[remember picture]\node[inner sep=0pt] (ul2) {};\al 7 & \dag 15 & \dag \al{11} \\
      \dag \al{12} & 15 & \dag{11} & \al  7 & 2\\
      15 & \dag \al{20} & 9 & 2 & 1 \tikz[remember picture]\node[inner sep=0pt] (br2) {};
    \end{array}  \right)
  \]
  \tikz[remember picture]\draw[overlay,dotted] ($ (ul2) + (-0.2, 0.45) $) rectangle ($ (br2) + (0.35, -0.2) $);

  The circled allocation $\Shares$ is not sequenceable: indeed, every
  sequence that could generate it should start with
  $\langle 3,2,\dots \rangle$, which leaves the frustrating sub-allocation
  $\Sharesp$ appearing in a dotted box above.

  Let us now choose an arbitrary agent who does not receive a top
  object in $\Sharesp$, for instance agent $a_1 = 2$. Let $o_1 = 3$ be
  her top object (of weight 11 in this case). The agent receiving
  $o_1$ in $\Sharesp$ is $a_2=1$. This agent prefers object $o_2= 4$
  (of weight 15), held by $a_1$, already encountered. We have built a
  cycle
  $ (a_1 , o_1) \rightarrow (a_{2} , o_{2}) \rightarrow (a_1, o_{1})$,
  in other words $(2 , 3) \rightarrow (1 , 4) \rightarrow (2, 3)$,
  that tells us exactly how to build another sub-allocation dominating
  $\Sharesp$. This sub-allocation can be built by replacing in
  $\Sharesp$ the attributions $ (a_1 \gets o_2 ) (a_{2} \gets o_1)$ by
  the attributions $ (a_1 \gets o_1) (a_{2} \gets o_{2})$. Hence, each
  agent involved in the cycle obtains a strictly better object than
  the previous one. Doing the same substitutions in the initial
  allocation $\Shares$ yields an allocation $\Shares'$ that dominates
  $\Shares$ (marked with $\dag$ in the matrix $\Poids$ above).
\end{example}

Now we will give the formal proof.\footnote{%
  This proof is similar to the proof of Brams and King
  \cite[Proposition 1, necessity]{BramsKing05}.  However, we give it
  entirely because it is more general, and because we will reuse it in
  the proof of our Proposition \ref{prop:CEEI}. Also note that the
  central idea of trading cycle is classical and is used, among
  others, by Varian \cite[page 79]{Varian1974} and Lipton \textit{et
    al.}  \cite[Lemma 2.2]{Lipton04} in the context of envy-freeness.}
\begin{proof} As stated in the example, we will now prove the
  contraposition of the proposition: every non-sequenceable allocation
  is dominated. Let $\Shares$ be a non-sequenceable allocation. From
  Proposition~\ref{prop:sousAlloc}, in a non-sequenceable allocation,
  there is at least one frustratring sub-allocation. Let $\Sharesp$ be
  such a sub-allocation (that can be $\Shares$ itself). We will, from
  $\Sharesp$, build another sub-allocation dominating it. Let us
  choose an arbitrary agent $a_1$ involved in $\Sharesp$, receiving an
  object not among her top ones in $\Sharesp$. Let $o_1$ be a top
  object of $a_1$ in $\Sharesp$, and let $a_2$ ($\neq a_1$) be the
  unique agent receiving it in $\Sharesp$. Let $o_2$ be a top object
  of $a_2$.  We can notice that $o_2 \neq o_1$ (otherwise $a_2$ would
  obtain one of her top objects and $\Sharesp$ would not be
  frustrating). Let $a_3$ be the unique agent receiving $o_2$ in
  $\Sharesp$, and so on.  Using this argument iteratively, we form a
  path starting from $a_1$ and alternating agents and objects, in
  which two successive agents and objects are distinct. Since the
  number of agents and objects is finite, we will eventually encounter
  an agent which has been encountered at a previous step of the
  path. Let $a_i$ be the first such agent and $o_k$ be the last object
  seen before her in the sequence ($a_i$ is the unique agent receiving
  $o_k$). We have built a cycle
  $ (a_i , o_i) \rightarrow (a_{i+1} , o_{i+1}) \cdots (a_k , o_{k})
  \rightarrow (a_i , o_i)$
  in which all the agents and objects are distinct, and that has at
  least two agents and two objects. From this cycle, we can modify
  $\Sharesp$ to build a new sub-allocation by giving to each agent in
  the cycle a top object instead of another less preferred object, all
  the agents not appearing in the cycle being left
  unchanged. More formally, the following attributions in $\Sharesp$ (and hence in $\Shares$):
  $ (a_i \gets o_k )  (a_{i+1} \gets o_i)  \cdots  (a_k \gets o_{k-1})$
  are replaced by:
  $ (a_i \gets o_i) (a_{i+1} \gets o_{i+1}) \cdots (a_k   \gets o_{k})$
  where $ (a \gets o ) $ means that $o$ is attributed to $a$. The same
  substitutions operated in $\Shares$ yield an allocation $\Shares'$
  that dominates $\Shares$.
\end{proof}

\begin{corollary} \label{cor:nonPO} No frustrating allocation can be
  Pareto-optimal (equivalently, in every Pareto-optimal allocation, at
  least one agent receives a top object).
\end{corollary}

Proposition~\ref{prop:POseq} implies that there exists, for a given
instance, three classes of allocations: (1) non-sequenceable
(therefore non Pareto-optimal) allocations, (2) sequenceable but non
Pareto-optimal allocations, and (3) Pareto-optimal (hence
sequenceable) allocations. These three classes define a ``scale of
efficiency'' that can be used to characterize the allocations. What is
interesting and new here is the intermediate level.


\section{Envy-Freeness and CEEI}
\label{sec:CEEI}

In the previous section we have investigated the link that exists between
efficiency and sequenceability. The use of sequence of sincere choices
can also be motivated by the search for a \emph{fair} allocation
protocol. Therefore, investigating the link between sequenceability
and fairness properties is a natural next step. In this section, we
will focus on two fairness properties, envy-freeness and competitive
equilibrium from equal income, and analyze their link with
sequenceability.

Envy-freeness \cite{Foley67,Varian1974} is probably one of the most
prominent fairness properties. It can be formally expressed in our
model as follows.

\begin{definition} \label{def:EF} Let $I$ be an add-MARA instance and
  $\Shares$ be an allocation.  $\Shares$ verifies the
  \emph{envy-freeness property} (or is simply \emph{envy-free}), when
  $\Ut_\ag(\share_\ag) \geq \Ut_\ag(\share_\agp)$,
  $\forall (\ag, \agp) \in \Agents^2$ (no agent strictly prefers the
  share of any other agent).
\end{definition}

The notion of \emph{competitive equilibrium} is an old and well-known
concept in economics \cite{Walras1874Elements,
  Fisher1891MathematicalInvestigations}. If equal incomes are imposed
among the stakeholders, this concept becomes the \emph{competitive
  equilibrium from equal incomes} \cite[page 177, for
example]{Moulin03}, yielding a very strong fairness concept which has
been recently explored in artificial intelligence
\cite{Othman2010,Budish11,Bouveret2015CharacterizingConflicts}. Here
is the definition of this concept adapted to our model.

\begin{definition}\label{def:CEEI}
  Let $I = (\Agents, \Objets, \poids)$ be an add-MARA instance, $\Shares$
  an allocation, and $\Prices \in [0, 1]^\maxObjets$ a vector of
  prices.  A pair $(\Shares, \Prices)$ is said to form a
  \emph{competitive equilibrium from equal incomes (CEEI)} if
  \[ 
    \forall \ag \in \Agents : \share_\ag \in \argmax_{\share \subseteq \Objets} \left\{
    \Ut_\ag(\share) :
    \sum_{\ob \in \share} \price_\ob \leq 1 \right\}.
  \]
  In other words, $\share_\ag$ is one of the maximal shares that
  $\ag$ can buy with a budget of 1, provided that the price of
  each object $\ob$ is $\price_\ob$.

  We will say that allocation $\Shares$ satisfies the CEEI test (is a
  CEEI allocation for short) if there exists a vector $\Prices$ such
  that $(\Shares, \Prices)$ forms a CEEI.
\end{definition}

As we have shown in a previous work
\cite{Bouveret2015CharacterizingConflicts}, every CEEI allocation is
envy-free in the model we use (additive, numerical preferences). In
this section, we investigate the question of whether an envy-free or
CEEI allocation is necessarily sequenceable.

For envy-freeness, the answer is negative.

\begin{prop}
  There exist non-sequenceable envy-free allocations, even if the
  agents' preferences are strict on shares.
\end{prop}

\begin{proof} A counterexample with strict preferences on shares is
  given in Example~\ref{ex:EFnonPOprefsStrictes} above, for which we
  can check that the circled allocation $\Shares$ is envy-free
  and non-sequenceable.
\end{proof}

Interestingly, for CEEI, however, the answer is positive.

\begin{prop} \label{prop:CEEI} Every CEEI allocation is
  sequenceable.
\end{prop}

It should be noted that we already know that every CEEI allocation is
Pareto-optimal if the preferences are strict on shares
\cite[Proposition 8]{Bouveret2015CharacterizingConflicts}. From
Observation~\ref{rem:str} and Proposition~\ref{prop:POseq}, if the
preferences are strict on shares, then every CEEI allocation is
sequenceable.  Proposition~\ref{prop:CEEI} is more general: no
assumption is made on the stricticity of preferences on shares (nor on
objects).  Note that a CEEI allocation can be ordinally necessary
Pareto-dominated, as the following example shows.
\[
\left(
  \begin{array}{cccc}
    \dag\al 2 & \dag 3 & 3 & \al 2\\
    2 & 3 & \dag\al 4 & 1 \\
    0 & \al 4 & 2 & \dag 4
  \end{array}
\right)
\]
The circled allocation is CEEI (with prices 0.5, 1, 1, 0.5) but is
ordinally necessary (hence also additively) dominated by the
allocation marked with $\dag$.

\begin{proof}  We will show that no allocation can be at the
  same time non-sequenceable and CEEI. Let $\Shares$ be a non-sequenceable
  allocation.  We can use the same terms and notations than in the
  proof of Proposition~\ref{prop:POseq}, especially concerning the
  dominance cycle.
  
  Let $\mathcal{C}$ be the set of agents concerned by the
  cycle. $\Shares$ contains the following shares:
\[    
    \share_{a_i} = \{ o_k \} \union \shareq_i \ \ \ \
    \share_{a_{i+1}} = \{ o_i \} \union \shareq_{i+1} \ \ \  \ .... \ \ \ \
    \share_{a_k} = \{ o_{k-1} \} \union \shareq_k 
\]
  whereas the allocation $\Shares'$ that dominates it, contains the
  following shares:
\[
    \share'_{a_i} = \{ o_i \} \union \shareq_i \ \ \ \
    \share'_{a_{i+1}} = \{ o_{i+1} \} \union \shareq_{i+1} \ \ \ \
                  .... \ \ \ \
    \share'_{a_k} = \{ o_{k} \} \union \shareq_k
\]
  the other shares being unchanged from $\Shares$ to $\Shares'$.

  Suppose that $\Shares$ is CEEI. This allocation must satisfy
  two kinds of constraints. First, $\Shares$ must satisfy the price
  constraint. If we write $p(\share) \Def \sum_{\ob\in\share} p_\ob$,
  \begin{equation} \label{eq:A}
    \forall \ag\in{\cal C} : p(\share_\ag) \leq 1
  \end{equation}
  Next, $\Shares$ must be optimal: every share having a higher utility
  for an agent than her share in $\Shares$ costs strictly more than
  1. Provided that
  $\forall \ag\in{\cal C} : \Ut_\ag(\share'_\ag) >
  \Ut_\ag(\share_\ag)$
  (because $\Shares'$ substitutes more preferred objects to less
  preferred objects in $\Shares$), this constraint can be written as
  \begin{equation} \label{eq:B}
    \forall \ag\in{\cal C} : p(\share'_\ag) > 1
  \end{equation}
  By summing equations~\ref{eq:A} and \ref{eq:B}, provided that all
  shares are disjoint, we obtain
  \begin{eqnarray*} \label{eq:C}
    p(\union_{\agp \in \cal C} \share_\agp) \leq |\cal C|
    & \text{and} &
                  p(\union_{\agp \in \cal C} \share'_\agp) > |\cal C|
  \end{eqnarray*}
  Yet,
  $\union_{\agp \in \cal C} \share_\agp=\union_{\ag \in \cal
    C}\share'_\agp$
  (because the allocation $\Shares'$ is obtained from $\Shares$ by
  simply swapping objects between agents in $\mathcal{C}$). The two
  previous equations are contradictory.
\end{proof}


\section{Experiments}
\label{sec:expe}

We have exhibited in Section~\ref{sec:Pareto}, a ``scale of allocation
efficiency'', made of three steps: non-sequenceable (NS), sequenceable
and non Pareto-optimal (SnP), and Pareto-optimal (PO). A natural
question is to know, for a given instance, which proportion of
allocations are located at each level of the scale. We give a first
answer in this section by experimentally characterizing the
distribution of allocations between the different levels. Moreover, we
analyze the relation between fairness and efficiency by linking this
scale of efficiency with the scale of fairness introduced in a
previous work \cite{Bouveret2015CharacterizingConflicts}. This scale
of fairness has six levels, from the weakest to the strongest
property: no criterion satisfied (--), maxmin share (MFS),
proportionality (PFS), minmax share (mFS), envy-freeness (EF) and
CEEI.

Our experimental protocol is the following. We have generated 100
add-MARA instances with 3 agents and 10 objects, reusing the uniform
and Gaussian random allocation protocol described in our previous work
\cite{Bouveret2015CharacterizingConflicts}. For each of these
instances, we have generated the set of possible allocations
($3^{10} = 59049$) and identified the highest level of fairness and
the level of efficiency of each one. We want to emphasize that a
particular difficulty was here to implement the CEEI test. This
problem has been proved to be \conp-hard \cite[Theorem
49]{Branzei2015ComputationalFair} and to the best of our knowledge no
practical method have been described yet, even for small numbers of
objects and agents. The method we have developed for this problem is
described in appendix and works in practice for problems up to 5
agents and 20 objects.

\begin{figure*}[htbp]
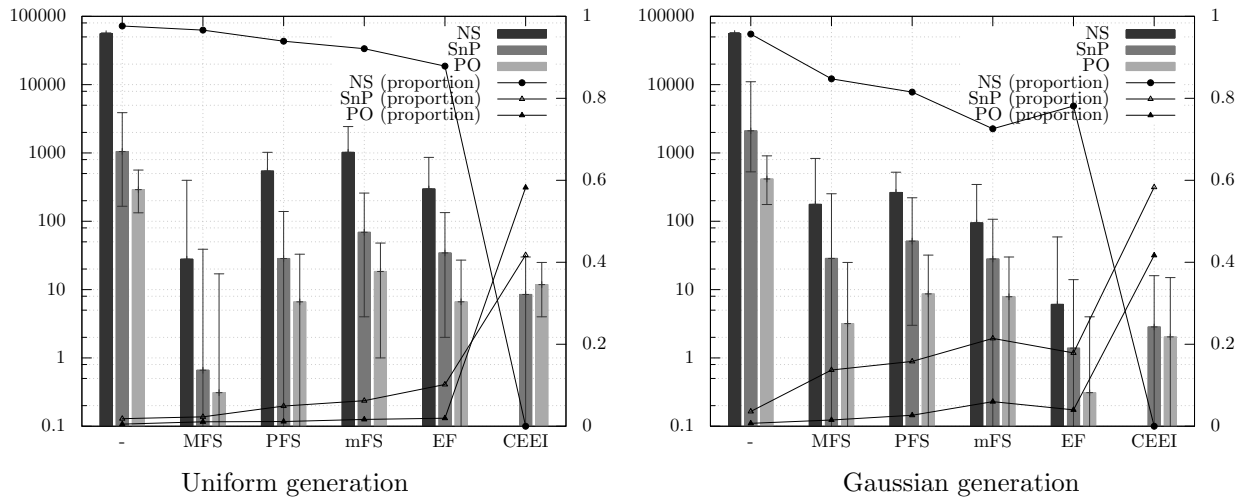

  {\parbox{0.48\textwidth}{
    \hspace{-1.5cm}\scalebox{0.7}{
      \input{graphe02}
    }\\{\hbox{}\hfill\hspace{-1.5cm} Uniform generation \hfill\hbox{}}
  }} \hfill
  \parbox{0.48\textwidth}{
    \hspace{-0.5cm}\scalebox{0.7}{
      \input{graphe01}
    }\\{\hbox{}\hfill\hspace{1.5cm} Gaussian generation \hfill\hbox{}}
  }
  \caption{Distribution of the number of allocations by pair of (efficiency, fairness) criteria.}
  \label{fig:graphe01}
\end{figure*}

Figure~\ref{fig:graphe01} gives a graphical overview of
the results.
In this figure, the histograms represent the average value (on all the
instances) of the set of allocations by pair of criteria (the min-max
interval has been represented with error bars), using a logarithmic
scale. The three curves represent the average proportion, by fairness
criterion, of the number of allocations satisfying the three
efficiency criteria, using a linear scale.

We can observe several interesting facts. First, a huge majority of
allocations do not have any efficiency nor fairness property (first
black bar on the left). Secondly, the distribution of the allocations
on the scale of efficiency seems to be correlated to the fairness
criteria: a higher proportion of sequenceable allocations can be found
among the envy-free allocations than among the allocations that do not
satisfy any fairness property, and for CEEI allocations and uniform
generation, there are even more Pareto-optimal allocations that just
sequenceable ones. Third, there is a higher proportion of allocations
satisfying fairness and efficiency properties for the uniform
generation model. Even if this seems logical for fairness (because a
``similar'' attraction for the objects --- as in the Gaussian model
--- is more likely to create conflicts among agents), the explanation
is not so clear for efficiency.


\section{Conclusion}
\label{sec:conclusion}

Sequences of sincere choices are arguably a remarkable protocol for
allocating a set of indivisible goods to agents. This protocol, known
for ages, has received a lot of interest in recent years by
researchers both in economics and computer science. In this paper,
following the work by Brams \textit{et al.}
\cite{BramsKing05}, we have shown that this
protocol, beyond being appealing in practice, also has a theoretical
interest in the context of numerical additive preferences. Namely, it
can be used to characterize the efficiency of an allocation by
defining an intermediate level between Pareto-optimality and no
efficiency at all. Moreover, we have introduced the simple notion of
frustrating (sub-)allocation and shown that it can be used to exactly
characterize the set of sequenceable allocations. We also have
characterized the set of instances for which there is an exact
one-to-one relation between allocations and sequences. Finally,
we have emphasized some links between fairness properties (especially
CEEI) and efficiency criteria. Although being technically simple, we
believe that these results are new and shed an interesting light on
sequences of sincere choices and Pareto-optimality.

This work opens up to some interesting questions, such as the impact
of restrictions on sequences like alternating sequences. Another
interesting topic is the relation between this protocol and social
welfare orderings, with questions such as the loss of social welfare
incurred by the execution of a sequence compared to the optimal
allocation (price of sequenceability).

\

\noindent\textbf{Acknowledgements}\quad We thank the anonymous
reviewers of AAMAS'16 for their valuable comments.



\bibliographystyle{plain}
\bibliography{partage}



\begin{contact}
  Sylvain Bouveret\\
  Laboratoire d'Informatique de Grenoble (LIG)\\
  Grenoble INP, Université Grenoble-Alpes\\
  Grenoble, France\\
  \email{sylvain.bouveret@imag.fr}
\end{contact}

\begin{contact}
  Michel Lema\^itre\\
  Formerly ONERA Toulouse, France\\
  \email{michel.lemaitre.31@gmail.com}
\end{contact}


\appendix

\newpage
\section*{Appendix}

\newcommand{\PRICE}{p'}
\newcommand{\PRICES}{\vect{\PRICE}}

\newcommand{\bp}{\bar \price}
\newcommand{\bP}{\bar \PRICE}

This appendix presents an exact (complete) and practical method to
decide whether a given allocation satisfies the CEEI test (Definition
\ref{def:CEEI}). We know that this problem is \conp-hard \cite[Theorem
49]{Branzei2015ComputationalFair}. The method we propose relies on
linear programming techniques. Our ambition is limited to be able to
tackle small instances (say no more than 5 agents and 20 objects) with
standard LP solvers.

The starting point is the following. The problem comes down to decide
if the following system of constraints is satisfiable. This system,
denoted by $S$, depends on the allocation $\Shares$ at stake. Its
variables are the object prices
$(\var{\price_\ob})_{\ob = 1}^\maxObjets$ and it comprises three sets
of constraints:

\begin{itemize}
\item
a first set of inequalities defining the domain of prices:
\begin{align*}
 0 \leq \var{\price_\ob} \leq 1,  \text{ for all } \ob \in \llb 1, \maxObjets \rrb;
\end{align*}
\item
 a second set of  $\maxAgents$ non strict inequalities,
modelling the fact that each agent's share is affordable:
\begin{align*}
\sum_{\ob = 1}^\maxObjets A_{\ag \ob} \var{\price_\ob} \leq 1, 
\text{ for all } \ag \in \llb 1, \maxAgents \rrb,
\end{align*}
with $A_{\ag \ob} = 1$ if $\ob \in \share_\ag$, 0 otherwise;
\item a third set of $q$ strict inequalities expressing the optimality
  of $\Shares$ given the prices: any better share for an agent costs
  strictly more than the given budget. Namely: for each agent $\ag$
  and each share $\share'$ such that
  $\Ut_\ag(\share') > \Ut_\ag(\share_\ag)$,
\begin{align*}
\sum_{\ob = 1}^\maxObjets B_{k \ob} \var{\price_\ob} > 1, 
\end{align*}
where $k$ is an index, different for each pair $(i, \share')$, and
$B_{k \ob} = 1$ if $\ob \in \share'$, 0 otherwise.
\end{itemize}

Note that these constraints cannot be properly handled by a (standard)
LP tool, because the inequalities of the third set are strict. The
Fourier-Motzkin elimination method
\cite{Fourier1826SolutionQuestion,Dantzig1973Fourier-Motzkin} could do
the job, as it can manage strict inequality constraints. However, its
drawback is its exponential number of steps.

Another possibility is to build, from $S$, a new system $S'$ of
non-strict linear inequalities which is equivalent to $S$ as far as
satisfiability is concerned. This system uses the same coefficients
$A_{\ag \ob} $ and $B_{k \ob}$, and is defined as follows.
 
\begin{itemize}
\item 
replace in each inequality each variable $\var{\price_\ob}$
by a new variable  $\var{\PRICE_\ob}$, with domains
\begin{align*}
 0 \leq  \var{\PRICE_\ob}  \text{ for all } \ob \in \llb 1, \maxObjets \rrb;
\end{align*}
\item 
replace in the second set of non strict inequalities, the budget bound 1
by a new variable $\var{d}$:
\begin{align*}
\sum_{\ob = 1}^\maxObjets A_{\ag \ob} \var{\PRICE_\ob} \leq \var{d}, 
\textrm{ for all } \ag \in \llb 1, \maxAgents \rrb;
\end{align*}
\item replace the third set of strict inequalities by the new
  following set of non-strict ones:
 \begin{align*}
\sum_{\ob = 1}^\maxObjets B_{k \ob} \var{\PRICE_\ob} \geq \var{d}+1, \text{ for all } k \in \llb 1, q \rrb,
\end{align*}
\end{itemize}

\begin{prop}
  $S$ is satisfiable if and only if $S'$ is satisfiable. Consequently,
  an allocation satisfies the CEEI test if and only if the
  corresponding system $S'$ is satisfiable. Moreover, in this case,
  giving to each object $\ob$ the price
  $\displaystyle \frac{ \bar{\PRICE_\ob}} {\bar{d} }$, with
  $\bar{\PRICE_\ob}$ and $\bar{d}$ the values of $\var{\PRICE_\ob}$
  and $\var{d}$ in the solution of $S'$ yields a CEEI with
  respect to the initial allocation $\Shares$.
\end{prop}

\begin{proof}
  ($\Longrightarrow$) Suppose $S$ is satisfiable.  Then there
  is a rational solution.\footnote{ Because coefficients are
    rational. It is a known result, which derives for example from the
    Fourier-Motzkin elimination method: the only required operations
    to build the solution (if it exists) are the standard arithmetic
    operators.}  Let $(\bp_\ob)_{\ob=1}^{\maxObjets}$ be this solution
  values, and let $\bar{d}$ be the least common multiple of the
  $\bp_\ob$ denominators. On the one hand we have:
  \begin{align*}
    \sum_{\ob = 1}^\maxObjets A_{\ag \ob} \bp_\ob \leq 1 
    \implies \sum_{\ob = 1}^\maxObjets A_{\ag \ob} \bar{d}  \bp_\ob \leq \bar{d}
  \end{align*}
  and on the other hand:
  \begin{align*}
    \sum_{\ob = 1}^\maxObjets B_{k \ob} \bp_\ob > 1 
    \implies \sum_{\ob = 1}^\maxObjets B_{k \ob} \bar{d}  \bp_\ob >  \bar{d} \\ 
    \implies \sum_{\ob = 1}^\maxObjets B_{k \ob} \bar{d}  \bp_\ob \geq  \bar{d} +1,
  \end{align*}
  because
  $\displaystyle \sum_{\ob = 1}^\maxObjets B_{k \ob} \bar{d} \bp_\ob $
  and $\bar{d}$ are integers.

  Hence $S'$ is satisfiable if $S$ is, with values $ \bar{d} \bp_\ob $
  for $\var{\PRICE_\ob}$ and $\bar{d}$ for $\var{d}$.

  \smallskip

  ($\Longleftarrow$)
  Suppose $S'$ satisfiable, and let $\bP_\ob$  and  $\bar{d}$ be the values of a solution.
  Then, on the one hand we have:
  \begin{align*}
    \sum_{\ob = 1}^\maxObjets A_{\ag \ob} \bP_\ob \leq  \bar{d}
    \implies \sum_{\ob = 1}^\maxObjets A_{\ag \ob} \frac{   \bP_\ob  } { \bar{d} }  \leq 1
  \end{align*}
  and on the other hand:
  \begin{align*}
    \sum_{\ob = 1}^\maxObjets B_{k \ob} \bar{\PRICE_\ob}  \geq  \bar{d} +1 
    & \implies \sum_{\ob = 1}^\maxObjets B_{k \ob}  \bar{\PRICE_\ob} >  \bar{d} \\
    & \implies \sum_{\ob = 1}^\maxObjets B_{k \ob} \frac{   \bar{\PRICE_\ob}  } { \bar{d} }  > 1
  \end{align*}

  Hence $S$ is satisfiable if $S'$ is,  
  with values $\displaystyle \frac {\bar{\PRICE_\ob}}  {\bar{d}}  $.
\end{proof}

The system $S'$ only has non-strict linear inequalities and a standard
LP tool can be used to test its satisfiability (there is no need for
integer linear programming). Note that it does not contradict the
\conp-hardness of the CEEI test since the number of constraints is not
polynomially bounded.

\smallskip

We now conclude with some practical remarks. It is known
\cite[Proposition 7]{Bouveret2015CharacterizingConflicts} that any
allocation satisfying the CEEI test is envy-free. Moreover, we know
from Proposition~\ref{prop:CEEI} that every CEEI allocation is
sequenceable. Since both properties can be tested in polynomial time
(envy-freeness can be tested $O(N^2M)$, sequenceability in $O(NM^2)$),
we can use them in a preliminary step to filter out allocations that
do not pass them. Finally, it can be noticed that the third set of
inequalities often contains a lot of redundant constraints that can be
eliminated to simplify the overall system.

\end{document}